\documentclass{llncs}

\usepackage[utf8]{inputenc} % allow utf-8 input
\usepackage[T1]{fontenc}    % use 8-bit T1 fonts
\usepackage{lmodern}
\usepackage{hyperref}       % hyperlinks
\usepackage{url}            % simple URL typesetting
\usepackage{booktabs}       % professional-quality tables
\usepackage{amsfonts}       % blackboard math symbols
\usepackage{nicefrac}       % compact symbols for 1/2, etc.
\usepackage{microtype}      % microtypography
\usepackage{amsmath}
\usepackage{amsfonts,amssymb}

\usepackage{float}
\newfloat{Algorithm}{hbtp}{lop}[section]
\usepackage{boxedminipage}

\newcommand{\veps}{\varepsilon}
\newcommand{\eps}{\varepsilon}
\newcommand{\MDA}{\mathsf{MinDisAgree}}
\newcommand{\MA}{\mathsf{MaxAgree}}

\newcommand{\mpr}{m^{\prime}}
\newcommand{\Vpr}{V^{\prime}}
\newcommand{\Epr}{E^{\prime}}

\newcommand{\OA}{\mathcal{A}}

\newcommand{\etal}{{\em et~al.}}

\title{Approximate Correlation Clustering Using Same-Cluster Queries}
\author{Nir Ailon\thanks{Nir Ailon acknowledges the generous support of ISF grant number 2021408}\inst{1} \and Anup Bhattacharya\thanks{Anup Bhattacharya acknowledges the support of TCS fellowship at IIT Delhi.}\inst{2} \and Ragesh Jaiswal\thanks{Ragesh Jaiswal acknowledges the support of ISF-UGC India-Israel Grant 2014.}\inst{2}}
\institute{
Technion, Haifa, Israel.\thanks{Email address: \email{nailon@cs.technion.ac.il}}
\and
Department of Computer Science and Engineering, \\
Indian Institute of Technology Delhi.\thanks{Email addresses: \email{\{anupb, rjaiswal\}@cse.iitd.ac.in}}
}

\date{}

\begin{document}
\maketitle

\begin{abstract} 
Ashtiani et al. (NIPS 2016) introduced a semi-supervised framework for clustering (SSAC) where a learner is allowed to make {\em same-cluster} queries. More specifically, in their model, there is a query oracle that answers queries of the form {``\em given any two vertices, do they belong to the same optimal cluster?''}. In many clustering contexts, this kind of oracle queries are feasible. Ashtiani et al. showed the usefulness of such a query framework by giving a polynomial time algorithm for the $k$-means clustering problem where the input dataset satisfies some {\em separation} condition. Ailon et al. extended the above work to the approximation setting by giving an efficient $(1+\veps)$-approximation algorithm for $k$-means for any small $\veps > 0$ and any dataset within the SSAC framework. In this work, we extend this line of study to the {\em correlation clustering} problem. Correlation clustering is a graph clustering problem where pairwise similarity (or dissimilarity) information is given for every pair of vertices and the objective is to partition the vertices into clusters that minimise the disagreement (or maximises agreement) with the pairwise information given as input. These problems are popularly known as $\mathsf{MinDisAgree}$ and $\mathsf{MaxAgree}$ problems, and $\mathsf{MinDisAgree}[k]$ and $\mathsf{MaxAgree}[k]$ are versions of these problems where the number of optimal clusters is at most $k$. There exist {\em Polynomial Time Approximation Schemes} (PTAS) for $\mathsf{MinDisAgree}[k]$ and $\mathsf{MaxAgree}[k]$ where the approximation guarantee is $(1+\veps)$ for any small $\veps$ and the running time is polynomial in the input parameters but exponential in $k$ and $1/\veps$. We get a significant running time improvement within the SSAC framework at the cost of making a small number of same-cluster queries. We obtain an $(1+ \veps)$-approximation algorithm for any small $\veps$ with running time that is polynomial in the input parameters and also in $k$ and $1/\veps$. We also give non-trivial upper and lower bounds on the number of same-cluster queries, the lower bound being based on the {\em Exponential Time Hypothesis} (ETH). Note that the existence of an efficient algorithm for $\mathsf{MinDisAgree}[k]$ in the SSAC setting exhibits the power of same-cluster queries since such polynomial time algorithm (polynomial even in $k$ and $1/\veps$) is not possible in the classical (non-query) setting due to our conditional lower bounds. Our conditional lower bound is particularly interesting as it not only establishes a lower bound on the number of same cluster queries in the SSAC framework but also establishes a conditional lower bound on the running time of any $(1+\veps)$-approximation algorithm for $\mathsf{MinDisAgree}[k]$.
\end{abstract}

\begin{section}{Introduction} 
Correlation clustering is a graph clustering problem where we are given similarity or dissimilarity information for pairs of vertices. The input is a graph $G$ on $n$ vertices. Edges of $G$ are labeled as similar (positive) or dissimilar (negative). The clustering objective is to partition the vertices into clusters such that edges labeled `positive' remain within clusters and `negative' edges go across clusters. However, this similarity/dissimilarity information may be inconsistent with this objective. For example, there may exist vertices $u, v, w$ such that edges $(u, v)$ and $(u, w)$ are labeled `positive' whereas edge $(v,w)$ is labeled `negative'. In this case, it is not possible to come up with a clustering of these $3$ vertices that would {\em agree} with all the edge labels. The objective of correlation clustering is to come up with a clustering that minimises disagreement or maximises agreement with the edge labels given as input. The minimisation version of the problem, known as $\mathsf{MinDisAgree}$, minimises the sum of the number of negative edges present inside clusters and the number of positive edges going across clusters. Similarly, the maximisation version is known as $\mathsf{MaxAgree}$ where the objective is to maximise the sum of the number of positive edges present inside clusters and the number of negative edges going across clusters. Unlike $k$-means or $k$-median clustering, in correlation clustering, there is no restriction on the number of clusters formed by the optimal clustering. When the number of optimal clusters is given to be at most $k$, these problems are known as $\mathsf{MinDisAgree}[k]$ and $\mathsf{MaxAgree}[k]$ respectively.

Bansal et al.~\cite{BBC2004} gave a constant approximation algorithm for $\MDA$ and a PTAS for $\MA$. Subsequently, Charikar et al.~\cite{CGW2005} improved approximation guarantee for $\MDA$ to $4$, and showed that $\MDA$ is $\mathsf{APX}$-hard. These results are for correlation clustering on complete graphs as it is known for general graphs, it is at least as hard as {\em minimum multi-cut} problem~\cite{CGW2005}. Since $\MDA$ is $\mathsf{APX}$-hard \cite{CGW2005}, additional assumptions were introduced for better results. For example \cite{MS2010,MMV2015} studied $\MDA$ where the input is noisy and comes from a semi-random model. When $k$ is given as part of the input, Giotis and Guruswami \cite{GG2006} gave a PTAS for $\MDA[k]$.

Recently there have been some works \cite{MM2016,Bl2016} with a {\em beyond-worst case} flavour where polynomial time algorithms for $\mathsf{NP}$-hard problems have been designed under some stability assumptions. Ashtiani et al.~\cite{AKB2016} considered one such stability assumption called {\em $\gamma$-margin}. They introduced a semi-supervised active learning (SSAC) framework and within this framework, gave a probabilistic polynomial time algorithm for $k$-means on datasets that satisfy the $\gamma$-margin property. More specifically, their SSAC framework involves a query oracle that answers queries of the form {\em ``given any two vertices, do they belong to the same optimal cluster?''}. The query oracle responds with a Yes/No answer where these answers are assumed to be consistent with some fixed optimal solution. In this framework, they studied the query complexity for polynomial time algorithms for $k$-means on datasets satisfying the $\gamma$-margin property. Ailon et al.~\cite{ABJK2017} extended this work to study query complexity bounds for $(1+\eps)$-approximation for $k$-means in SSAC framework for any small $\eps>0$ without any stability assumption on the dataset. They gave almost matching upper and lower bounds on the number of queries for $(1+\veps)$-approximation of $k$-means problem in SSAC framework. 

In this work, we study $\MDA[k]$ in the SSAC framework, where the optimal clustering has at most $k$ clusters and give upper and lower bounds on the number of same-cluster queries for $(1+\eps)$-approximation for correlation clustering for any $\eps>0$. We also give upper bounds for $\MA[k]$. Our algorithm is based on the PTAS by Giotis and Guruswami \cite{GG2006} for $\MDA[k]$. The algorithm by Giotis and Guruswami involves random sampling a subset $S$ of vertices and considers all possible ways of partitioning $S$ into $k$ clusters $S=\{S_1,\ldots,S_k\}$, and for every such $k$-partitioning, clusters the rest of the vertices greedily. Every vertex $v \in V\setminus S$ is assigned a cluster $S_j$ that maximizes its agreement with the edge labels. Their main result was the following.

\begin{theorem}[Giotis and Guruswami~\cite{GG2006}] 
For every $k\geq 2$, there is a PTAS for $\MDA[k]$ with running time $n^{O(9^k/\eps^2)}\log n$.
\end{theorem}

Since Giotis and Guruswami considered all possible ways of partitioning subset $S$ into $k$ clusters, their running time has exponential dependence on $k$. Here, we make the simple observation that within the SSAC framework we can overcome this exponential dependence on $k$ by making same-cluster queries to the oracle. The basic idea is to randomly sample a subset $S$ of vertices as before and partition it optimally into $k$ clusters by making same-cluster queries to the oracle. Note that by making at most $k|S|$ same-cluster queries, one can partition $S$ optimally into $k$ clusters. Once we have subset $S$ partitioned as in the optimal clustering (a key step needed in the analysis of Giotis and Guruswami) we follow their algorithm and analysis for $(1+\eps)$-approximation for $\MDA[k]$. Here is our main result for $\MDA[k]$ in the SSAC framework. We obtain similar results for $\MA[k]$.

\begin{theorem}[Main result: Upper bound]\label{thm:main-upper} Let $\veps > 0$ and $k \geq 2$. There is a {\em (randomized)} algorithm in the SSAC framework for $\MDA[k]$ that uses $O\left( \frac{k^{14}\log k\log n}{\eps^6} \right)$ same-cluster queries, runs in time $O(\frac{n k^{14}\log k\log n}{\eps^6})$ and outputs a $(1 + \veps)$-approximate solution with high probability.
\end{theorem}

We complement our upper bound result by providing a lower bound on the number of queries in the SSAC framework for any efficient $(1+\eps)$-approximation algorithm for $\MDA$ for any $\eps>0$. Our lower bound result is conditioned on the {\em Exponential Time Hypothesis} (ETH hypothesis) \cite{IP01,IPZ01}. Our lower bound result implies that the number of queries is depended on the number of optimal clusters $k$. Our main result with respect to query lower bound is given as follows.

\begin{theorem}[Main result: Lower bound]\label{thm:main-lower} Given that Exponential Time Hypothesis {\em (ETH)} holds, there exists a constant $\delta>0$ such that any $(1+\delta)$-approximation algorithm for $\MDA[k]$ in the SSAC framework that runs in polynomial time makes $\Omega(\frac{k}{\text{poly} \log{k}})$ same-cluster queries.
%requires $2^{\Omega(\frac{k}{\text{poly}\log k})}$ time where $k$ is the number of optimal clusters.
\end{theorem}

Exponential Time Hypothesis is the following statement regarding the hardness of the $3$-$\mathsf{SAT}$ problem.

\begin{quote}
{\em Exponential Time Hypothesis (ETH)\cite{IP01,IPZ01}}: There does not exist an algorithm that can decide whether any $3$-$\mathsf{SAT}$ formula with $m$ clauses is satisfiable with running time $2^{o(m)}$.
\end{quote}

Note that our query lower bound result is a simple corollary of the following theorem that we prove.

\begin{theorem} If the Exponential Time Hypothesis (ETH) holds, then there exists a constant $\delta > 0$ such that any $(1+\delta)$-approximation algorithm for $\MDA[k]$ requires $2^{\Omega(\frac{k}{poly \log{k}})}$ time.
\end{theorem}

The above lower bound statement may be of independent interest. 
It was already known that $\MDA$ is $\mathsf{APX}$-hard.
Our result is a non-trivial addition to the understanding of the hardness of the correlation clustering problem.
Given that our query upper bound result is through making simple observations in the algorithms of Giotis and Guruswami, our lower bound results may be regarded as the primary contribution of this work.
So, we first give our lower bound results in the next section and the upper bound results in Section~\ref{sec-upper}.
However, before we start discussing our results, here is a brief discussion on the related works.

\begin{paragraph}{Related Works} There have been numerous works on clustering problems in semi-supervised settings. Balcan and Blum \cite{BB2008} proposed an interactive framework for clustering which use `split/merge' queries. In this framework, given any abritrary clustering $C=\{C_1,C_2,\ldots,\}$ as query, oracle specifies some cluster $C_l$ should be split or clusters $C_i$ and $C_j$ should be merged. Awasthi et al.~\cite{ABV2014} developed a local clustering algorithm which uses these split/merge queries. {\em One versus all} queries for clustering were studied by Voevodski et al.~\cite{VBRTX2014}. The oracle, on a query $s\in X$, returns distances from $s$ to all points in $X$. The authors provided a clustering, close to optimal $k$-median clustering, with only $O(k)$ such queries on instances satisfying $(c,\eps)$-approximation stability property \cite{BBG2013}. Fomin \etal~\cite{FKPPV2014} gave a conditional lower bound for the {\em cluster~editing} problem which can also be stated as a decision version of the correlation clustering problem. In the $p$-cluster editing problem, given a graph $G$ and a budget $B$, and an integer $p$, the objective is to decide whether $G$ can be transformed into a union of $p$ clusters (disjoint cliques) using at most $B$ edge additions and deletions. Assuming ETH, they showed that there exists $p=\Theta(k^{\omega})$ for some $0\leq \omega \leq 1$ such that there is no algorithm that decides in time $2^{o(\sqrt{pB})} \cdot n^{O(1)}$ whether $G$ can be transformed into a union of $p$ cliques using at most $B$ adjustments (edge additions and deletions). It is not clear whether their exact reduction can be modified into an approximation preserving reduction to obtain results similar to what we have here. Mazumdar and Saha \cite{MS2017} studied correlation clustering problem in a similar setting where edge similarity and dissimilarity information are assumed to be coming from two distributions. Given such an input, they studied the {\em cluster recovery} problem in SSAC framework, and gave upper and lower bounds on the query complexity. Their lower bound results are information theoretic in nature. We are, however, interested in the approximate solutions for the correlation clustering problem.\end{paragraph}

%Our upper and lower bound results are proved in Section \ref{sec-upper} and Section \ref{sec-lower} respectively. 
\end{section}

\section{Query Lower Bounds } \label{sec-lower} 
In this section, we obtain a lower bound on the number of same-cluster queries that any FPTAS within the SSAC framework needs to make for the problem $\MDA[k]$. 
We derive a conditional lower bound for the minimum number of queries under the Exponential Time Hypothesis (ETH) assumption. 
Some such conditional lower bound results based on ETH can be found in \cite{M16}. 
We prove the following main theorem in this section.

\begin{theorem}\label{thm:l1}
If the Exponential Time Hypothesis (ETH) holds, then there exists a constant $\delta>0$ such that any $(1+\delta)$-approximation algorithm for $\MDA[k]$ requires $2^{\Omega(\frac{k}{\text{poly}\log k})}$ time.
\end{theorem}

The above theorem gives a proof of Theorem~\ref{thm:main-lower}.

\begin{proof}[Proof of Theorem~\ref{thm:main-lower}]
Let us assume that there exists a query-FPTAS that makes only $o(\frac{k}{\text{poly}\log k})$ same-cluster queries. Then, by considering all possible answers for these queries and picking the best solution, one can solve the problem in $2^{o(\frac{k}{\text{poly}\log k})}$ time which contradicts Theorem~\ref{thm:l1}. \end{proof}

In the remaining section, we give the proof of Theorem~\ref{thm:l1}.
First, we state the ETH hypothesis. 
Our lower bound results are derived assuming this hypothesis.

\begin{quote} 
\underline{\bf Hypothesis 1} {\it (Exponential Time Hypothesis (ETH)\cite{IP01,IPZ01})}: There does not exist an algorithm that decides whether any $3$-$\mathsf{SAT}$ formula with $m$ clauses is satisfiable with running time $2^{o(m)}$.
\end{quote}

Since we would like to obtain lower bounds in the approximation domain, we will need a {\em gap} version of the above ETH hypothesis. 
The following version of the PCP theorem would be very useful in obtaining a gap version of ETH.

\begin{theorem}[Dinur's PCP Theorem~\cite{D07}] \label{lemma-pcp} 
For some constants $\veps,d>0$, there exists a polynomial-time reduction that takes a $3$-$\mathsf{SAT}$ formula $\psi$ with $m$ clauses as input and produces one $\mathsf{E}3$-$\mathsf{SAT}$\footnote{Every clause in an $\mathsf{E}3$-$\mathsf{SAT}$ formula has exactly $3$ literals.} 
formula $\phi$ with $\mpr=O(m\text{poly}\log m)$ clauses such that
\begin{itemize}
	\item if $\psi$ is satisfiable, then $\phi$ is satisfiable, and
	\item if $\psi$ is unsatisfiable, then $\text{val}(\phi)\leq 1-\eps$, and
	\item each variable in $\phi$ appears in at most $d$ clauses.
\end{itemize}
where $\text{val}(\phi)$ is the maximum fraction of clauses of $\phi$ which are satisfiable by any assignment.
\end{theorem}

The hypothesis below follows from ETH and the above Theorem \ref{lemma-pcp}, and will be useful for our analysis.

\begin{quote} 
\underline{\bf Hypothesis 2}: There exists constants $\eps,d>0$ such that the following holds: There does not exist an algorithm that, given a $\mathsf{E}3$-$\mathsf{SAT}$ formula $\psi$ with $m$ clauses and each variable appearing in at most $d$ clauses, distinguishes whether $\psi$ is satisfiable or $val(\psi)\leq (1-\eps)$, and runs in time better than $2^{\Omega \left(\frac{m}{\text{poly} \log m} \right)}$.
\end{quote}

The lemma given below trivially follows from Dinur's PCP Theorem \ref{lemma-pcp}.

\begin{lemma}\label{lemma:H1toH2}
If Hypothesis 1 holds, then so does Hypothesis 2.
\end{lemma}

We now give a reduction from the gap version of the $\mathsf{E}3$-$\mathsf{SAT}$ problem to the gap version of the $\mathsf{NAE}3$-$\mathsf{SAT}$ problem.
A problem instance of $\mathsf{NAE}3$-$\mathsf{SAT}$ consists of a set of clauses (each containing exactly $3$ literals) and a clause is said to be satisfied by an assignment iff at least one and at most two literals in the clause is true ($\mathsf{NAE}$ stands for ``Not All Equal").
For any instance $\phi$, we define $val'(\phi)$ to be the maximum fraction of clauses that can be satisfied in the ``not all equal" sense by an assignment. 
Note that this is different from $val(\phi)$ which is equal to the maximum fraction of clauses that can be satisfied (in the usual sense).
First, we reduce $\mathsf{E}3$-$\mathsf{SAT}$ to $\mathsf{NAE}6$-$\mathsf{SAT}$ and then $\mathsf{NAE}6$-$\mathsf{SAT}$ to $\mathsf{NAE}3$-$\mathsf{SAT}$. 
%The proofs for most of the lemmas in this section are skipped due to space limitations and the complete discussion (including all proofs) may be found in the supplementary material. 

\begin{lemma} \label{esat-nae6sat} 
Let $0 < \veps < 1$ and $d > 1$.
There is a polynomial time reduction that given an instance $\psi$ of $\mathsf{E}3$-$\mathsf{SAT}$ with $m$ clauses with each variable appearing in at most $d$ clauses, produces an instance $\phi$ of $\mathsf{NAE}6$-$\mathsf{SAT}$ with $4m$ clauses such that
\begin{enumerate}
	\item If $val(\psi)=1$, then $val'(\phi) = 1$, and
	\item If $val(\psi) \leq (1-\veps)$, then $val'(\phi) \leq (1 - \veps/4)$, and
	\item Each variable in $\phi$ appears in at most $4d$ clauses.
\end{enumerate}
\end{lemma}

\begin{proof}
We construct $\phi$ in the following manner: for every variable $x_i$ in $\psi$, we introduce two variables $y_i$ and $z_i$. We will use $x_i = 1$ iff $y_i \neq z_i$ for every $i$ in our reduction.
For every clause $(l_i, l_j, l_k)$ (with $l_i, l_j, l_k$ being literals), we introduce the following four NAE clauses in $\phi$:
$$(p_i, q_i, p_j, q_j, p_k, q_k), (p_i, q_i, p_j, q_j, \bar{p}_k, \bar{q}_k), (p_i, q_i, \bar{p}_j, \bar{q}_j, p_k, q_k), (p_i, q_i, \bar{p}_j, \bar{q}_j, \bar{p}_k, \bar{q}_k)$$
For any index (say $i$), if $l_i = x_i$ (that is, the variable is in the positive form), then $p_i = y_i$ and $q_i = z_i$.
On the other hand, if $l_i = \bar{x}_i$, then $p_i = y_i$ and $q_i = \bar{z}_i$.
So for example, for the clause $(x_2, \bar{x}_7, x_9)$ in $\psi$, we have the following four clauses:
$$(y_2, z_2, y_7, \bar{z}_7, y_9, z_9), (y_2, z_2, y_7, \bar{z}_7, \bar{y}_9, \bar{z}_9), (y_2, z_2, \bar{y}_7, z_7, y_9, z_9), (y_2, z_2, \bar{y}_7, z_7, \bar{y}_9, \bar{z}_9)$$
Note that property (3) of the lemma holds due to our construction.
For property (1), we argue that for any satisfying assignment for $\psi$, the assignment of variables in $\phi$ as per the rule $x_i = 1$ iff $y_i \neq z_i$ is a satisfying assignment of $\psi$ (in the NAE sense).
This is because for every literal $l$ that makes a clause in $\psi$ true, the two corresponding copies $p$ and $q$ satisfies all the four clauses (in the NAE sense).
For property (2), we prove the contrapositive.
Suppose there is an assignment to the $y, z$ variables in $\phi$ that satisfies at least $(1 - \veps/4)$ fraction of the clauses.
We will argue that the assignment to the variables of $\psi$ as per the rule $x_i = 1$ iff $y_i \neq z_i$ satisfies at least $(1 - \veps)$ fraction of clauses of $\psi$.
First, note that for every set of $4$ clauses in $\phi$ created from a single clause of $\psi$, either $3$ of them are satisfied or all four are satisfied (whatever the assignment of the variable be).
Let $m_1$ be the number of these 4-sets where all 4 clauses are satisfied and let $m_2$ be the number of these 4-sets where 3 clauses are satisfied, where $m=m_1+m_2$. Then we have $4 m_1 + 3 m_2 \geq (1 - \veps/4)\cdot (4m)$ which implies that $m_1 \geq (1 - \veps)m$.
Note that for any of the 4-sets where all 4 clauses are satisfied, the corresponding clause in $\psi$ is satisfied with respect to the assignment as per rule $x_i =1$ iff $y_i \neq z_i$ (since at least one the $p, q$ pairs will have opposite values).
So, the fraction of the clauses satisfied in $\psi$ is at least $\frac{m_1}{m} \geq (1 - \veps)$.
\end{proof}

\begin{lemma} \label{nae6sat-nae3sat} 
Let $0 < \veps < 1$ and $d > 1$.
There is a polynomial time reduction that given an instance $\psi$ of $\mathsf{NAE}6$-$\mathsf{SAT}$ with $m$ clauses  and with each variable appearing in at most $d$ clauses, produces an instance $\phi$ of $\mathsf{NAE}3$-$\mathsf{SAT}$ with $4m$ clauses such that:
\begin{enumerate}
\item If $val'(\psi)=1$, then $val'(\phi) = 1$, and
\item If $val'(\psi) \leq (1-\veps)$, then $val'(\phi) \leq (1 - \veps/4)$.
\item Each variable in $\phi$ appears in at most $\max{(d, 2)}$ clauses.
\end{enumerate}
\end{lemma}

\begin{proof}
For every clause $C_i = (a_i, b_i, c_i, d_i, e_i, f_i)$ in $\psi$, we construct the following four clauses in $\phi$ (let us call it a 4-set): $(a_i, b_i, x_i),(\bar{x}_i, c_i, y_i),(\bar{y}_i, d_i, z_i), (\bar{z}_i, e_i, f_i)$, introducing new variables $x_i, y_i, z_i$. Property (3) trivially holds for this construction. For every satisfying assignment for $\psi$, there is a way to set the clause variables $x_i, y_i, z_i$ for every $i$ such that all four clauses in the 4-set corresponding to clause $C_i$ are satisfied. So, property (1) holds. We show property (2) using contraposition. Consider any assignment of $\phi$ that satisfies at least $(1 - \veps/4)$ fraction of the clauses. Let $m_j$ denote the number of 4-sets such that as per this assignment $j$ out of $4$ clauses are satisfied. Then, we have $\sum_{j=0}^{4} j \cdot m_j \geq (1 - \veps/4) \cdot (4m)$. This implies that: $3 \cdot \sum_{j=0}^{3} m_j + 4 m_4 \geq (1 - \veps/4) \cdot (4m)$ which implies that $m_4 \geq (1 - \veps) m$. Now, note that for any 4-set such that all four clauses are satisfied, the corresponding clause in $\psi$ is satisfied by the same assignment to the variables. This implies that there is an assignment that makes at least $(1 - \veps)$ fraction of clauses true in $\psi$.
\end{proof}

We come up with the following hypothesis which holds given that Hypothesis 2 holds, and is crucial for our analysis.

\begin{quote} 
\underline{\bf Hypothesis 3}: There exists constants $\veps,d>0$ such that the following holds: There does not exist an algorithm that, given a $\mathsf{NAE}3$-$\mathsf{SAT}$ formula $\psi$ with $m$ clauses with each variable appearing in at most $d$ clauses, distinguishes whether $val'(\psi) = 1$ or $val'(\psi)\leq (1-\veps)$, and runs in time better than $2^{\Omega \left(\frac{m}{\text{poly} \log m} \right)}$.
\end{quote}

The lemma given below follows easily from the Lemmas \ref{esat-nae6sat} and \ref{nae6sat-nae3sat} above.
\begin{lemma}\label{lemma:H2toH3}
If Hypothesis 2 holds, then so does Hypothesis 3.
\end{lemma}

We now give a reduction from the gap version of $\mathsf{NAE}3$-$\mathsf{SAT}$ to the gap version of {\em monotone} $\mathsf{NAE}3$-$\mathsf{SAT}$ that has no negative variables. Note that because of the NAE (not all equal) property, setting all variables to $1$ does not necessarily satisfy the formula.
%Again, the proof is skipped and may be found in the supplementary material.

\begin{lemma} \label{nae3sat-mononae3sat} Let $0 < \veps < 1$ and $d > 1$. There is a polynomial time reduction that given an instance $\psi$ of $\mathsf{NAE}3$-$\mathsf{SAT}$ with $m$ clauses and with each variable appearing in at most $d$ clauses, produces an instance $\phi$ of {\em monotone} $\mathsf{NAE}3$-$\mathsf{SAT}$ with $O(m)$ clauses such that:
\begin{enumerate}
\item If $val'(\psi)=1$, then $val'(\phi) = 1$, and
\item If $val'(\psi) \leq (1-\veps)$, then $val'(\phi) \leq (1 - \frac{\veps}{1+12d})$.
\item Each variable in $\phi$ appears in at most $4d$ clauses.
\end{enumerate}
\end{lemma}

\begin{proof}
We construct $\phi$ in the following manner: Substitute all positive literals of the variable $x_i$ with $y_i$ and all negative literals with $z_i$ for new variables $y_i, z_i$. Also, for every variable $x_i$, add the following $4d$ clauses:
$$\{(y_i, z_i, t_i^j), (y_i, z_i, u_i^j), (y_i, z_i, v_i^j), (t_i^j, u_i^j, v_i^j)\}_{j=1}^{d}$$
where $t_i^j, u_i^j, v_i^j$ for $1 \leq j \leq d$ are new variables. Note that the only way to satisfy all the above clauses is to have $y_i \neq z_i$.
Let $m'$ denote the total number of clauses in $\phi$.
So, $m' = m + 4dn$. Also, from the construction, each variable in $\phi$ appears in at most $4d$ clauses. This proves property (3).
Property (1) follows from the fact that for any satisfying assignment for $\psi$, there is a way to extend this assignment to variables in $\phi$ such that all clauses are satisfied.
For all $i$, $y_i = x_i$ and $z_i =  \bar{x}_i$. All the new variables $t, u, v$ can be set so as to make all the new clauses satisfied.

We argue property (2) using contraposition.
Suppose there is an assignment to variables in $\phi$ that makes at least $(1 - \veps/(1+12d))$ fraction of clauses satisfied.
First, note that there is also an assignment that makes at least $(1 - \veps/(1+12d))$ fraction of the clauses satisfied and in which for all $i$, $y_i \neq z_i$. This is because $3d$ out of $4d$ of the following clauses can be satisfied when $y_i = z_i$:
$$\{(y_i, z_i, t_i^j), (y_i, z_i, u_i^j), (y_i, z_i, v_i^j), (t_i^j, u_i^j, v_i^j)\}_{j=1}^{d}$$
However, if we flip one of $y_i, z_i$, then the number of above clauses satisfied can be $4d$ and we might lose out on at most $d$ clauses since a variable appears in at most $d$ clauses in $\psi$.
Let $m'$ be the number of clauses corresponding to the original clauses that are satisfied with this assignment. So, we have $m' + 4nd > (1 - \frac{\veps}{(1+12d)})(m+4nd)$ which gives:
\begin{eqnarray*}
m' > (1 - \veps) m  + \frac{12md\veps}{1 + 12d}  - \frac{4nd\veps}{1+12d} \geq (1 - \veps)m \quad \textrm{(since $3m \geq n$)}
\end{eqnarray*}
This completes the proof of the lemma.
\end{proof}

We come up with the following hypothesis which holds given that Hypothesis 3 holds.
\begin{quote}
\underline{\bf Hypothesis 4}: There exists constants $\eps,d>0$ such that the following holds: There does not exist an algorithm that, given a monotone $\mathsf{NAE}3$-$\mathsf{SAT}$ formula $\psi$ with $m$ clauses with each variable appearing in at most $d$ clauses, distinguishes whether $val'(\psi) = 1$ or $val'(\psi)\leq (1-\eps)$, and runs in time better than $2^{\Omega \left(\frac{m}{\text{poly} \log m} \right)}$.
\end{quote}

The lemma below follows easily from Lemma \ref{nae3sat-mononae3sat} mentioned in above.
\begin{lemma}\label{lemma:H3toH4}
If Hypothesis 3 holds, then so does Hypothesis 4.
\end{lemma}

We provide a reduction from the gap version of monotone $\mathsf{NAE}3$-$\mathsf{SAT}$ to a gap version of $2$-colorability of $3$-uniform bounded degree hypergraph. 

\begin{lemma} \label{mononae3sat-2colorhyper} 
Let $0 < \veps < 1$ and $d > 1$.
There exists a polynomial time reduction that given a monotone $\mathsf{NAE}3$-$\mathsf{SAT}$ instance $\psi$ with $m$ clauses and with every variable appearing in at most $d$ clauses, outputs an instance $H$ of $3$-uniform hypergraph with $O(m)$ vertices and hyperedges and with bounded degree $d$ such that if $\psi$ is satisfiable, then $H$ is $2$-colorable, and if at most $(1-\eps)$-fraction of clauses of $\psi$ are satisfiable, then any $2$-coloring of $H$ would have at most $(1-\eps)$-fraction of edges that are bichromatic.
\end{lemma}

\begin{proof} The reduction constructs a hypergraph $H(V,E)$ as follows.
The set of vertices $V$ correspond to the set of variables (all of them positive literals) of the monotone $\mathsf{NAE}3$-$\mathsf{SAT}$ instance $\psi$.
The set of edges $E$ correspond to the set of clauses all of which have $3$ literals, and therefore every hyperedge is of size $3$.
The resulting hypergraph is $3$-uniform, and since every variable appears in at most $d$ clauses, the hypergraph $H$ is of bounded degree $d$, and $|V|=O(m)$ and $|E|=O(m)$.
If there exists a satisfying assignment for $\psi$, then every edge in $H$ is bichromatic and the hypergraph would be $2$-colorable, and if at most $(1-\eps)$-fraction of clauses are satisfiable by any assignment, then at most $(1-\eps)$-fraction of edges of $H$ are bichromatic.
\end{proof}

We come up with the following hypothesis which holds given that Hypothesis 4 holds.
\begin{quote}
\underline{\bf Hypothesis 5}: There exists constants $\eps,d>0$ such that the following holds: There does not exist an algorithm that, given a $3$-uniform hypergraph $H$ with $m$ vertices and where every vertex has degree at most $d$, distinguishes whether $H$ is bichromatic or at most $(1-\eps)$-fraction of edges are bichromatic, and runs in time better than $2^{\Omega \left(\frac{m}{\text{poly} \log m} \right)}$.
\end{quote}

The lemma below follows easily from Lemma \ref{mononae3sat-2colorhyper} above.
\begin{lemma}\label{lemma:H4toH5}
If Hypothesis 4 holds, then so does Hypothesis 5.
\end{lemma}

We now give a reduction from $2$-colorability in $3$-uniform hypergraph $H$ with constant bounded degree to a correlation clustering instance on a complete graph $G$. We use the reduction as given in \cite{CGW2005} for our purposes.

\begin{lemma}[\cite{CGW2005}] \label{hyper-cor} 
Let $\veps, d > 0$.
There is a polynomial-time reduction that given a $3$-uniform hypergraph $H(V,E)$ with $m$ vertices and where each vertex appears in at most $d$ hyperedges, outputs an instance of the correlation clustering problem where the graph $G(\Vpr,\Epr)$ has $N=O(m)$ vertices and $M=2N$ edges with edges in $\Epr$ are labeled as `positive' and all the other edges in the complete graph on $\Vpr$ vertices are labeled as `negative' such that the following holds:
\begin{enumerate}
\item If $H$ is 2-colorable, then the cost of the optimal correlation clustering is $M - N$, and

\item If at most $(1 - \veps)$-fraction of hyperedges of $H$ are bi-chromatic, then the optimal cost of correlation clustering is at least $M - (1 - \delta)N$, where $\delta$ is some constant.
\end{enumerate}
\end{lemma}

We come up with the following hypothesis which holds given that Hypothesis 5 holds.
\begin{quote}
\underline{\bf Hypothesis 6}: There exists constants $\eps>0$ such that the following holds: There does not exist a $(1+\veps)$-factor approximation algorithm for the $\MDA[k]$ problem that runs in time better than $poly(n) \cdot 2^{\Omega(\frac{k}{poly \log{k}})}$.
\end{quote}

The lemma below follows easily from Lemma \ref{hyper-cor} given above.
\begin{lemma}\label{lemma:H5toH6}
If Hypothesis 5 holds, then so does Hypothesis 6.
\end{lemma}

Finally, the proof of Theorem~\ref{thm:l1} follows from chaining together lemmas \ref{lemma:H1toH2}, \ref{lemma:H2toH3}, \ref{lemma:H3toH4}, \ref{lemma:H4toH5}, and \ref{lemma:H5toH6}.

\section{Algorithms for $\MA[k]$ and $\MDA[k]$ in SSAC Framework} \label{sec-upper}

In this section, we give $(1+\eps)$-approximation algorithms for the $\MA[k]$ and $\MDA[k]$ problems within the SSAC framework for any $\eps>0$.

\subsection{$\MA[k]$}
In this section, we will discuss a query algorithm that gives $(1+\veps)$-approximation to the $\MA[k]$ problem. The algorithm that we will discuss is closely related to the non-query algorithm for $\MA[k]$ by Giotis and Guruswami. See Algorithm {\bf MaxAg($k, \veps$)} in \cite{GG2006}. In fact, except for a few changes, this section will look extremely similar to Section 3 in \cite{GG2006}. Given this, it will help if we mention the high-level idea of the Giotis-Guruswami algorithm and point out the changes that can be made within the SSAC framework to obtain the desired result. The algorithm of Giotis and Guruswami proceeds in $m$ iterations, where $m = O(1/\veps)$. The given dataset $V$ is partitioned into $m$ equal parts $V^1,\ldots,V^m$, and in the $i^{th}$ iteration, points in $V^i$ are assigned to one of the $k$ clusters. In order to cluster $V^i$ in the $i^{th}$ iteration, the algorithm samples a set of data points $S^i$, and for all possible $k$-partitions of $S^i$, it checks the agreement of a point $v \in V^i$ with the $k$ clusters of $S^i$. Suppose for a particular clustering $S^i_1,\ldots,S^i_k$ of $S^i$, the agreement of vertices in $V^i$ is maximised. Then the vertices in $V^i$ are clustered by placing them into the cluster that maximises their agreement with respect to $S^i_1,\ldots,S^i_k$. Trying out all possible $k$-partitions of $S^i$ is an expensive operation in the Giotis-Guruswami algorithm (since the running time becomes $\Omega(k^{|S^i|})$). This is where the same-cluster queries help. Instead of trying out all possible $k$-partitions of $S^i$, we can make use of the same-cluster queries to find a {\em single} appropriate $k$-partition of $S^i$ in the $i^{th}$ iteration. This is the clustering that matches the ``hybrid'' clustering of Giotis and Guruswami. So, the running time of the $i^{th}$ iteration improves from $O(k^{|S^i|})$ to $O(k \cdot |S^i|)$. Moreover, the number of same-cluster queries made in the $i^{th}$ iteration is $k \cdot |S^i|$, thus making the total number of same-cluster queries to be $O(\frac{k}{\veps} \cdot |S^i|)$. The theorem is given below. The details of the proof of this theorem is not given since it trivially follows from Giotis and Guruswami (see Theorem 3.2 in \cite{GG2006}).

\begin{theorem}
There is a query algorithm $\mathsf{QueryMaxAg}$ that behaves as follows: On input $\veps, \delta$ and a labelling $\mathcal{L}$ of the edges of a complete graph $G$ with $n$ vertices, with probability at least $(1-\delta)$, algorithm $\mathsf{QueryMaxAg}$ outputs a $k$ clustering of the graph such that the number of agreements induced by this $k$-clustering is at least $OPT-\veps n^2/2$, where $OPT$ is the optimal number of agreements induced by any $k$-clustering of $G$. The running time of the algorithm is $O \left( \frac{nk}{\veps^3} \log{\frac{k}{\veps^2 \delta}}\right)$. Moreover, the number of same-cluster queries made by $\mathsf{QueryMaxAg}$ is $O \left( \frac{k}{\veps^3} \log{\frac{k}{\veps^2 \delta}}\right)$.
\end{theorem}

Using the simple observation that $OPT \geq n^2/16$ (see proof of Theorem 3.1 in \cite{GG2006}), we get that the above query algorithm gives $(1+\veps)$-approximation guarantee in the SSAC framework.

\subsection{$\MDA[k]$} 

In this section, we provide a $(1+\veps)$-approximation algorithm for the $\MDA[k]$ for any small $\veps>0$. Giotis and Guruswami \cite{GG2006} provided a $(1 + \veps)$-approximation algorithm for $\MDA[k]$. In this work, we extend their algorithm to make it work in the SSAC framework with the aid of same-cluster queries, and thereby improve the running time of the algorithm considerably. Our query algorithm will be closely based on the non-query algorithm of Giotis and Guruswami. In fact, except for a small (but crucial) change, the algorithms are the same. So, we begin by discussing the main ideas and the result by Giotis and Guruswami.

\begin{lemma}[Theorem 4.7 in \cite{GG2006}] 
For every $k\geq 2$ and $\veps > 0$, there is a $(1+\veps)$-approximation algorithm for $\MDA[k]$ with running time $n^{O(9^k/\veps^2)}\log n$.
\end{lemma}

The algorithm by Giotis and Guruswami builds on the following ideas. First, from the discussion in the previous section, we know that there is a FPTAS within SSAC framework for $\MA[k]$. Therefore, unless $OPT$, the optimal value for $\MDA[k]$ is small (OPT$=\gamma n^2$, for some small $\gamma>0$), the complement solution for $\MA[k]$ would give a valid $(1+\veps)$-approximate solution for $\MDA[k]$. 
Since $OPT$ is small, this implies that the optimal value for $\MA[k]$ is large which means that for any random vertex $v$ in graph $G=(V,E)$, a lot of edges incident on $v$ agree to the optimal clustering. Suppose we are given a random subset $S\subseteq V$ of vertices that are optimally clustered $S=\{S_1,\ldots,S_k\}$, and let us assume that $S$ is sufficiently large. 
Since most of the edges in $E$ are in agreement with the optimal clustering, we would be able to 
assign vertices in $V \setminus S$ to their respective clusters greedily.
For any arbitrary $v\in V\setminus S$, assign $v$ to $S_i$ for which the number of edges that agree is maximized. 
Giotis and Guruswami observed that clustering vertices in $V\setminus S$ in this manner would work with high probability when these vertices belong to large clusters. For vertices in small clusters, we may not be able to decide assignments with high probability. They carry out this greedy assignment of vertices in $V\setminus S$ into clusters $S_1,\ldots,S_k$, and filter out clusters that are sufficiently large and recursively run the same procedure on the union of small clusters.

For any randomly sampled subset $S\subseteq V$ of vertices, Giotis and Guruswami try out all possible ways of partitioning $S$ into $k$ clusters in order to partition $S$ optimally into $k$ clusters $S_1,\ldots,S_k$.
This ensures that at least one of the partitions matches the optimal partition. However, this exhaustive way of partitioning imposes huge burden on the running time of the algorithm. In fact, their algorithm runs in $n^{O(9^k/\veps^2)}\log n$ time. Using access to the same-cluster query oracle, we can obtain a significant reduction in the running time of the algorithm. 
We query the oracle with pairs of vertices in $S$ and since query answers are assumed to be consistent with some unique optimal solution, optimal $k$ clustering of vertices in $S$ is accomplished using at most $k|S|$ same-cluster queries. 
Once we have a $k$-partitioning of sample $S$ that is consistent with the optimal $k$-clusters, we follow the remaining steps of \cite{GG2006}.
The same approximation analysis as in \cite{GG2006} follows for the query algorithm.
For completeness we give the modified algorithm in Figure~\ref{alg1}. Let oracle $\OA$ take any two vertices $u$ and $v$ as input and return `Yes' if they belong to the same cluster in optimal clustering, and `No' otherwise. 

\vspace*{-0.1in}

\begin{center}
\begin{Algorithm}
\begin{boxedminipage}{\textwidth}
$\mathsf{QueryMinDisAgree(k,\alpha)}$

Input: Labeling $\mathcal{L}:{n\choose 2} \rightarrow \{+,-\}$ of edges of graph $G(V,E)$, Oracle $\OA:{n\choose 2}\rightarrow \{\text{Yes,No}\}$.\\
Output: A $k$-clustering of $G$\\
Constants: $c_1=\frac{1}{20}$\\ \newline
(1) If $k=1$, return $1$-clustering.\\
(2) Run $\mathsf{QueryMaxAg}$ on input $\mathcal{L}$ with accuracy $\frac{\alpha^2c_1^2}{32k^4}$ to obtain $k$-clustering $ClusMax$.\\
(3) Set $\beta=\frac{c_1 \alpha}{16k^2}$. Pick sample $S\subseteq V$ of size $\frac{5\log n}{\beta^2}$ independently and uniformly at random with replacement from $V$.\\
(4) Optimally cluster $S=\{S_1,\ldots,S_k\}$ by making same-cluster queries to oracle $\OA$.\\
(5) Let $C_j=S_j$ for $1\leq j\leq k$.\\
(6) For each $u\in V\setminus S$\\
\hspace*{0.1in} (6.1) $\forall i=1,\ldots,k$: Let $l^u_i$ be the number of edges which agree between $u$ and nodes in $S_i$.\\
\hspace*{0.1in}	(6.2) Let $j_u=\arg \max_i l^u_i$ be the index of the cluster which maximizes the above quantity.\\
\hspace*{0.1in}	(6.3) $C_{j_u}=C_{j_u}\cup \{u\}$.\\
(7) Let the set of large and small clusters be $Large=\{j:1\leq j\leq k,|S_j|\geq \frac{n}{2k}\}$ and $Small=[k]\setminus \text{Large}$. \\
(8) Let $l=|Large|$, and $s=k-l$.\\
(9) Cluster $W=\cup_{j\in \text{Small}} S_j$ into $s$ clusters $\{W_1,\ldots,W_s\}$ using recursive calls to $\mathsf{QueryMinDisAgree(s,\alpha)}$.\\
(10) Let $ClusMin$ be clustering obtained by $k$ clusters $\{S_j\}_{j\in \text{Large}}$ and $\{W_t\}_{\{1\leq t\leq s\}}$.\\
(11) Return the better of $ClusMin$ and $ClusMax$
\end{boxedminipage}
\caption{Query version of the algorithm by Giotis and Guruswami.}
\label{alg1}
\end{Algorithm}
\end{center}

\vspace*{-0.35in}

Here is the main theorem giving approximation guarantee of the above algorithm. As stated earlier, the proof follows from the proof of a similar theorem (Theorem 4.7 in \cite{GG2006}) by Giotis and Guruswami.

\begin{theorem}
Let $0 < \veps \leq 1/2$. For any input labelling, $\mathsf{QueryMinDisAgree(k, \veps/4)}$ returns a $k$-clustering with the number of disagreements within a factor of $(1 + \veps)$ of the optimal.
\footnote{Readers familiar with \cite{GG2006} will realise that the statement of the theorem is slightly different from statement of the similar theorem (Theorem 13) in \cite{GG2006}. More specifically, the claim is about the function call with $\veps/4$ as a parameter rather than $\veps$. This is done to allow the recursive call in step (9) to be made with same value of precision parameter as the initial call. This does not change the approximation analysis but is crucial for our running time analysis.}
\end{theorem}

Even though the approximation analysis of the query algorithm remains the same as the non-query algorithm of Giotis and Guruswami, the running time analysis changes significantly.
Let us write a recurrence relation for the running time of our recursive algorithm. Let $T(k)$ denote the running time of the algorithm when $n$ node graph is supposed to be clustered into $k$ clusters with a given precision parameter $\alpha$. Using the results of the previous subsection, the running time of step (2) is $O(\frac{nk^{13}\log k}{\alpha^6})$. The running time for partitioning the set $S$ is given by $k |S|$ which is $O(\frac{k^5 \log{n}}{\alpha^2})$. Steps (6-8) would cost $O(n k |S|)$ time which is $O(\frac{nk^5 \log{n}}{\alpha^2})$. 
So, the recurrence relation for the running time may be written as $T(k) = T(k-1) + O(\frac{nk^{13} \log k \log{n}}{\alpha^6})$. This simplifies to $O(\frac{nk^{14}\log k\log{n}}{\alpha^6})$. 
As far as the same-cluster queries are concerned, we can write a similar recurrence relation.
$Q(k) = Q(k-1) + O(\frac{k^{13}\log k \log{n}}{\alpha^6})$ which simplifies to $O(\frac{k^{14}\log k \log{n}}{\alpha^6})$.
This completes the proof of Theorem~\ref{thm:main-upper}.

\section{Conclusion and Open Problems} 
In this work, we give upper and lower bounds on the number of same-cluster queries to obtain an efficient  $(1+\eps)$-approximation for correlation clustering on complete graphs in the SSAC framework of Ashtiani et al.
Our lower bound results are based on the Exponential Time Hypothesis (ETH). 
It is an interesting open problem to give unconditional lower bounds for these problems. Another interesting open problem is to design query based algorithms with faulty oracles. 
This setting is more practical since in many contexts it may not be known whether any two vertices belong to the same optimal cluster with high confidence. 
Mitzenmacher and Tsourakakis \cite{MT2016} designed a query based algorithm for clustering where query oracle, similar to our model, answers whether any two vertices belong to the same optimal cluster or otherwise but these answers may be wrong with some probability less than $\frac{1}{2}$. 
For $k=2$, we can use their algorithm to obtain $(1+\eps)$-approximation for $\MDA[k]$ with faulty oracle. 
However, for $k>2$ they needed stronger query model to obtain good clusterings. 
Designing an efficient $(1+\eps)$-approximation algorithm for $\MDA[k]$ with faulty oracle is an interesting open problem.

\bibliographystyle{plain}
\bibliography{biblio}
\end{document}